\documentclass[11pt]{article}
\usepackage{amsmath}
\usepackage{amsthm}
\usepackage{amssymb}
\usepackage{a4wide}
\usepackage{todonotes}
\usepackage{microtype}
\usepackage{mathrsfs}
\usepackage{authblk}

\newtheorem{theorem}{Theorem}

\newtheorem{definition}[theorem]{Definition}

\newcommand{\hp}{\ensuremath{\mathrm{hp}}}
\newcommand{\Sym}{\ensuremath{\mathrm{Sym}}}
\newcommand{\sgn}{\ensuremath{\mathrm{sgn}}}
\newcommand{\inin}{\ensuremath{\mathrm{in}}}
\newcommand{\outin}{\ensuremath{\mathrm{out}}}

\title{Directed Hamiltonicity and Out-Branchings \\ via Generalized Laplacians}
\author[1]{Andreas Bj\"orklund  }
\author[2]{Petteri Kaski}
\author[3]{Ioannis Koutis}
\affil[1]{\small Department of Computer Science, Lund University,  andreas.bjorklund@yahoo.se}
\affil[2]{Department of Computer Science, Aalto University, petteri.kaski@aalto.fi }
\affil[3]{Computer Science Department, UPR Rio Piedras, ioannis.koutis@upr.edu}
\date{}

\begin{document}

\maketitle 

\setlength\parindent{0pt} 

\renewcommand{\labelenumi}{\alph{enumi}.} 

\begin{abstract}
\noindent
We are motivated by a tantalizing open question in exact algorithms: can we detect whether an $n$-vertex directed graph $G$ has a Hamiltonian cycle in time significantly less than $2^n$?

\smallskip
\noindent We present new randomized algorithms that improve upon several previous works:
\smallskip 
\begin{enumerate}
\item
We show that for any constant $0<\lambda<1$ and prime $p$ we can \textbf{count} the Hamiltonian cycles modulo $p^{\lfloor (1-\lambda)\frac{n}{3p}\rfloor}$ in expected time 
less than $c^n$ for a constant $c<2$ that depends only on $p$ and $\lambda$.
Such an algorithm was previously known only for the case of 
counting modulo {\em two}~[Bj\"orklund and Husfeldt,~FOCS~2013].
\item We show that we can detect a Hamiltonian cycle in  $O^*(3^{n-\alpha(G)})$ time and \textbf{polynomial space}, where $\alpha(G)$ is the size of the maximum independent set in $G$. In particular, this yields an $O^*(3^{n/2})$ time algorithm for bipartite directed graphs, which is faster than the exponential-space algorithm in [Cygan~{\em et~al.},~STOC~2013]. 
\end{enumerate}
\smallskip
Our algorithms are based on the algebraic combinatorics of ``incidence assignments'' that we can capture through evaluation of determinants of Laplacian-like matrices, inspired by the Matrix--Tree Theorem for directed graphs. In addition to the novel algorithms for directed Hamiltonicity, we use the Matrix--Tree Theorem to derive simple algebraic algorithms for detecting out-branchings. Specifically, we give an $O^*(2^k)$-time randomized algorithm for detecting out-branchings with at least $k$ internal vertices, improving upon the algorithms of [Zehavi, ESA~2015] and [Bj\"orklund~{\em et~al.}, ICALP~2015]. We also present an algebraic algorithm for the directed $k$-Leaf problem, based on a non-standard monomial detection problem.

\end{abstract}

\section{Introduction}

The Hamiltonian cycle problem 
has played a prominent role in development of techniques for the
design of {\em exact algorithms} for hard problems.
The early $O^*(2^n)$ algorithms based on dynamic programming
and inclusion-exclusion~\cite{Bellman1962,KGK77,Karp1982},
remained un-challenged for several decades. In 2010,
Bj\"orklund~\cite{determinant2014bjorklund},
gave a randomized algorithm running in $O(1.657^n)$ time
for the case of undirected graphs. The algorithm
taps into the power of algebraic combinatorics, and in particular
determinants that enumerate cycle covers. 

Despite this progress in the undirected
Hamiltonian cycle problem, a substantial improvement in the more general directed
version of the problem remains an open problem
and a key challenge in the area of exact algorithms. 
The currently best known general algorithm
runs in $O^*(2^{n-\Theta(\sqrt{n/\log n})})$ time~\cite{B16}, and there are no known
connections with the theory of SETH-hardness~\cite{Impagliazzo2001}
that would---at least partly---dash the hope for a faster algorithm. 

A number of recent works have attempted to crack directed Hamiltonicity,
revealing that the problem is indeed easier in certain restricted settings. 
Cygan and Pilipczuk~\cite{faster2013cygan} showed
that the problem admits an $O^*(2^{(1-\epsilon_d)n})$ time
algorithm  for graphs with average degree bounded by $d$,
where $\epsilon_d$ is a constant with a doubly exponential
dependence on~$d$.
Cygan {\em et~al.}~\cite{fast2014cygan} showed that the problem admits 
an $O^*(1.888^n)$ time randomized algorithm for bipartite graphs and that 
the parity of directed Hamiltonian cycles can also be
computed within the same time bound. Bj\"orklund and 
Husfeldt~\cite{parity2013bjorklund} showed
that the parity of Hamiltonian cycles can be computed in $O^*(1.619^n)$ 
randomized time
in general directed graphs. Finally, 
Bj\"orklund~{\em et~al.}~\cite{Bjorklund2015}
showed that the problem can be solved in $O^*((2-\Theta(1))^n)$ time when
the graph contains less than $1.038^n$ Hamiltonian cycles, via a reduction
to the parity problem. In this paper we improve or generalize all of these 
works. 

\subparagraph*{Our results.} 
As one would expect, all recent
``below-$2^n$'' algorithm designs for the Hamiltonicity problem 
rely on algebraic combinatorics and involve 
formulas that enumerate Hamiltonian cycles. But somewhat surprisingly, 
none of these approaches employs the directed version of the Matrix--Tree 
Theorem~(see e.g.~Gessel and Stanley~\cite[\S11]{Gessel1995}), one of 
the most striking and beautiful results in algebraic graph theory.
The theorem enables the enumeration of spanning out-branchings, 
that is, rooted spanning trees with all arcs oriented away from the root, 
via a determinant polynomial. Our results in this paper derive from 
a detailed combinatorial understanding and generalization of this 
classical setup.

The combinatorial protagonist of this paper is the following notion
that enables a ``two-way'' possibility to view each arc in a directed graph:

\begin{definition}[Incidence assignment]
Let $G$ be a directed graph with vertex set $V$ and arc set $E$. 
For a subset $W\subseteq V$ we say that a mapping $\mu:W\rightarrow E$ 
is an {\em incidence assignment} if for all $u\in W$ it holds that 
$\mu(u)$ is incident with $u$.
\end{definition}
In particular, looking at a single arc $uv\in E$, an incidence assignment 
$\mu$ can assign $uv$ in two%
\footnote{Strictly speaking we are here assuming that both $u\in W$
and $v\in W$. To break symmetry in our applications we do allow also 
situations where $uv$ has only one possible assignment due to 
either $u\notin W$ or $v\notin W$.}{}
possible ways: as an out-arc $\mu(u)=uv$ at $u$, 
or as an in-arc $\mu(v)=uv$ at $v$. 

From an enumeration perspective 
the serendipity of this ``two-way'' possibility to assign an arc becomes 
apparent when one considers how an incidence assignment $\mu$ can realize 
a directed cycle in its image $\mu(W)$. Indeed, let 
\[
u_1u_2,\ u_2u_3,\ \ldots,\ u_{\ell-1}u_{\ell},\ u_{\ell}u_1\in E
\]
be the arcs of a directed cycle $C$ of length $\ell\geq 2$ in $G$
with $V(C)\subseteq W$. It is immediate that there are now exactly two%
\footnote{Again strictly speaking it will be serendipitous to break 
symmetry so that certain cycles will have only one realization 
instead of two.}{} 
ways to realize $C$ in the image $\mu(W)$. Namely, we can realize $C$  
either (i) using only in-arcs with 
\begin{equation}
\label{eq:in-realization}
\mu(u_1)=u_{\ell}u_1,\quad
\mu(u_2)=u_1u_2,\quad
\mu(u_3)=u_2u_3,\quad
\ldots,\quad
\mu(u_{\ell})=u_{\ell-1}u_{\ell}\,,
\end{equation}
or (ii) using only out-arcs with 
\begin{equation}
\label{eq:out-realization}
\mu(u_1)=u_1u_2,\quad
\mu(u_2)=u_2u_3,\quad
\mu(u_3)=u_3u_4,\quad
\ldots,\quad
\mu(u_{\ell})=u_{\ell}u_1\,.
\end{equation}
Incidence assignments thus enable two distinct ways to realize 
a {\em directed} cycle. Furthermore, it is possible to {\em switch} 
between \eqref{eq:in-realization} and \eqref{eq:out-realization} so that
only the images of $u_1,u_2,\ldots,u_{\ell}$ under $\mu$ are affected.
The algebraization of this combinatorial observation is at the heart
of the directed Matrix--Tree Theorem (which we will review for convenience
of exposition in Sect.~\ref{sect:laplacian}) and all of our results in 
this paper. 

Our warmup result involves a generalization of the directed Hamiltonian
path problem, namely the $k$-{\em Internal Out-Branching} problem, 
where the goal is to detect whether a given directed 
graph contains a spanning out-branching that has at least $k$
internal vertices. This is a well-studied problem on its own, with several 
successive improvements the latest of which is an $O^*(3.617^k)$ algorithm
by Zehavi~\cite{Zehavi14} and an $O^*(3.455^k)$
algorithm by Bj\"orklund~{\em et~al.}~\cite{DBLP:conf/icalp/BjorklundKKZ15} 
for the undirected version of the problem. 

Using a combination of the directed Matrix--Tree Theorem and a 
monomial-sieving idea due to Floderus~{\em et~al.}~\cite{FloderusLPS15}, 
in Sect.~\ref{sect:branchings} we show the following:

\begin{theorem}[Detecting a $k$-Internal Out-Branching]
	\label{thm:k-internal-out-branchings}
	There exists a randomized algorithm that solves the $k$-internal 
	out-branching problem in time $O^*(2^k)$ and with negligible
	probability of reporting a false negative.
\end{theorem}
In Appendix~\ref{appendix:k-leaf} we give a further application for the $k$-Leaf problem, 
that is, detecting a spanning out-branching with at least $k$ leaves. 
We note that Gabizon~{\em et~al.}~\cite{Gabizon2015} 
have recently given another application of the directed Matrix--Tree Theorem 
for the problem of detecting out-branchings of bounded degree. 

Proceeding to our two main results, 
in Sect.~\ref{sect:modular} we observe that the directed Matrix--Tree Theorem 
leads to a formula for computing the number of Hamiltonian paths in arbitrary
characteristic by using a standard inclusion--exclusion approach, 
which leads to a formula that involves the summation of $2^n$ determinants. 
To obtain a below-$2^n$ design, we present a way to randomize the underlying 
Laplacian matrix so that the number of Hamiltonian paths does not change 
but in expectation most of the summands vanish modulo a prime power.
Furthermore, to efficiently list the non-vanishing terms, we use a variation 
of an algorithm of Bj\"orklund~{\em et~al.}~\cite{BHL15} that was used for 
a related problem, computing the permanent modulo a prime power. 
This leads to our first main result:

\begin{theorem} [Counting directed Hamiltonian cycles modulo a prime power]
    \label{thm:counting-modulo-prime-power}
For all $0<\lambda<1$ there exists a randomized algorithm that,
given an $n$-vertex directed graph and a prime $p$ as input,
counts the number of Hamiltonian cycles modulo 
$p^{\lfloor (1-\lambda)n/(3p)\rfloor}$ in expected time 
$O^*\bigl(2^{n(1-\lambda^2/(19p\log_2 p))}\bigr)$.
The algorithm uses exponential space.
\end{theorem}

A corollary of Theorem~\ref{thm:counting-modulo-prime-power}
is that if $G$ has at most $d^n$ Hamiltonian cycles, we can detect
one in time $O(c_d^n)$, where $d$ is any fixed constant
and $c_d<2$ is a constant that only depends on $d$. As a further 
corollary we obtain a randomized algorithm for counting Hamiltonian cycles 
in graphs of bounded average (out-)degree $d$ in $O(2^{(1-\epsilon_d)n})$ 
time. The constant $\epsilon_d$ has a polynomial dependency in $d$. 
Previous algorithms had a constant $\epsilon_d$ with 
an exponential dependency on $d$~\cite{traveling2012bjorklund,faster2013cygan}.
(The proofs of these results are relegated to Appendix~\ref{appendix:counting-modulo-prime-power}.)

Returning to undirected Hamiltonicity, a key to the algorithm 
in~\cite{determinant2014bjorklund} was the observation that determinants
enumerate all non-trivial cycle covers an even number of times. This is 
due to the fact that each undirected cycle can be traversed in both directions. 
By picking a special vertex, one can break symmetry
and force this to happen only for non-Hamiltonian cycle covers,
so that the corresponding monomials cancel in characteristic 2.
In Sect.~\ref{sect:quasi-laplacian} we present
a ``quasi-Laplacian'' matrix whose determinant
enables a similar approach for the directed case via algebraic
combinatorics of incidence assignments, and furthermore enables one
to accommodate a speedup assuming the existence of a good-sized 
independent set. We specifically prove the following as our second main
result:

\begin{theorem}[Detecting a directed Hamiltonian cycle]
	\label{thm:directed-hamiltonicity}
	There exists a randomized algorithm that solves the directed Hamiltonian cycle 
	problem on a given directed graph $G$ with a maximum independent set of size $\alpha(G)$, in 
	$O^*(3^{n(G)-\alpha(G)})$ time, polynomial space and with negligible probability
	of reporting a false negative.
\end{theorem}
Theorem~\ref{thm:directed-hamiltonicity} improves and generalizes 
the exponential-space algorithm of~Cygan~{\em et~al.}~\cite{fast2014cygan}.

\subparagraph*{Terminology and conventions.}

All graphs in this paper are directed and without loops and parallel arcs
unless indicated otherwise. 
For an arc $e$ starting from vertex $u$ and ending at vertex $v$ we say that $u$ is the {\em tail} of $e$ and $v$ is the {\em head} of $e$. 
The vertices $u$ and $v$ are the {\em ends} of $e$. 
A directed graph is {\em connected} if the 
undirected graph obtained by removing orientation from the arcs is connected. 
A subgraph of a graph is {\em spanning} if the subgraph has the same set 
of vertices as the graph. A connected directed graph is an {\em out-branching}
if every vertex has in-degree $1$ except for the {\em root} vertex that has
in-degree $0$. We say that a vertex is {\em internal} to an out-branching 
if it has out-degree at least $1$; otherwise the vertex is a {\em leaf} of the
out-branching. The ({\em directed}) {\em Hamiltonian cycle} problem 
asks, given a directed graph $G$ as input, whether $G$ has a spanning 
directed cycle as a subgraph. The notation $O^*(\ )$ suppresses a 
multiplicative factor polynomial in the input size. We say that an
event parameterized by $n$ has {\em negligible} probability if 
the probability of the event tends to zero as $n$ grows without bound.

\section{The symbolic Laplacian of a directed graph}

\label{sect:laplacian}

This section develops the relevant preliminaries on directed graph Laplacians.

\subparagraph*{Permutations and the determinant.}
A bijection $\sigma:U\rightarrow U$ of a finite set $U$ is called
a {\em permutation} of $U$. A permutation $\sigma$ {\em moves} an
element $u\in U$ if $\sigma(u)\neq u$; otherwise $\sigma$ {\em fixes} $u$.
The {\em identity} permutation fixes every element of $U$. 
A permutation $\sigma$ of $U$ is a 
{\em cycle} of {\em length} $k\geq 2$ if there exist distinct 
$u_1,u_2,\ldots,u_k\in U$ with 
$\sigma(u_1)=u_2,\,\sigma(u_2)=u_3\,,\ldots,\,\sigma(u_{k-1})=u_k,\,\sigma(u_k)=u_1$
and $\sigma$ fixes all other elements of $U$. Two cycles are {\em disjoint}
if every point moved by one is fixed by the other. 
The set of all permutations of $U$ forms
the {\em symmetric group} $\Sym(U)$ with the composition of mappings 
as the product operation of the group. Every nonidentity permutation factors 
into a unique product of pairwise disjoint cycles. The {\em sign} of a 
permutation $\sigma$ that factors into $c$ disjoint cycles of lengths
$k_1,k_2,\ldots,k_c$ is $\sgn(\sigma)=(-1)^{\sum_{j=1}^c (k_j-1)}$. 
The sign of the identity permutation is $1$.

The {\em determinant} of a square matrix $A$ with rows and columns indexed
by $U$ is the multivariate polynomial
\[
\det A=\sum_{\sigma\in\Sym(U)}\sgn(\sigma)\prod_{u\in U}a_{u,\sigma(u)}\,.
\]

\subparagraph*{The punctured Laplacian determinant via incidence assignments.}
Let $G$ be a directed graph with $n$ vertices. 
Associate with each arc $uv\in E=E(G)$ an indeterminate $x_{uv}$. 
The {\em symbolic Laplacian} $L=L(G)$ of $G$ is the $n\times n$ matrix
with rows and columns indexed by the vertices $u,v\in V=V(G)$ and
the $(u,v)$-entry defined%
\footnote{Recall that we assume that $G$ is loopless so the entries 
with $u=v$ are well-defined.}{}
by
\begin{equation}
\label{eq:laplacian}
\ell_{uv}=
\begin{cases}
\sum_{w\in V:wu\in E} x_{wu} & \text{if $u=v$};\\
-x_{uv}                      & \text{if $uv\in E$};\\
0                            & \text{if $u\neq v$ and $uv\notin E$}.
\end{cases}
\end{equation}
Observe that for each $v\in V$ we have that column $v$ of $L$ sums to
zero because the diagonal entries cancel the negative off-diagonal entries.
Furthermore, for each $u\in V$ we have that the monomials on 
row $u$ of $L$ correspond to the arcs incident to $u$. 
Indeed, each monomial at the diagonal corresponds to an
in-arc to $u$, and each monomial at an off-diagonal entry corresponds
to an out-arc from $u$. Thus, selecting one monomial from each row 
corresponds to selecting an incidence assignment.

To break symmetry, select an $r\in V$. 
The symbolic Laplacian of $G$ {\em punctured at} $r$ is 
obtained from $L$ by deleting both row $r$ and column $r$. 
We write $L_r=L_r(G)$ for the symbolic Laplacian of $G$ punctured 
at $r$. Let us write $\mathscr{B}_r=\mathscr{B}_r(G)$ for the set of all 
spanning out-branchings of $G$ with root $r\in V$. The following theorem is 
well-known (see e.g.~Gessel and Stanley~\cite[\S11]{Gessel1995}) and is
presented here for purposes of displaying a proof that presents the
cancellation argument using incidence assignments.

\begin{theorem}[Directed Matrix--Tree Theorem]
\label{thm:multivariate-out-branching}
$\det L_r=\sum_{H\in\mathscr{B}_r}\prod_{uv\in E(H)} x_{uv}$.
\end{theorem}
\begin{proof}
Let us abbreviate $V_r=V(G)\setminus\{r\}$ and study the determinant
\begin{equation}
\label{eq:det-lr}
\det L_r=\sum_{\sigma\in\Sym(V_r)}\sgn(\sigma)\prod_{u\in V_r}\ell_{u,\sigma(u)}\,.
\end{equation}
In particular, let us fix an arbitrary permutation $\sigma\in\Sym(V_r)$
and study the monomials of the polynomial $\prod_{u\in V_r}\ell_{u,\sigma(u)}$
with the assumption that this polynomial is nonzero. 
From \eqref{eq:laplacian} it is immediate for each $u\in V_r$ 
that $\ell_{u,\sigma(u)}$ expands either (i) to the diagonal sum 
$\sum_{w\in V:wu\in E} x_{wu}$, which happens precisely when $\sigma$ 
fixes $u$ with $\sigma(u)=u$, or (ii) to the off-diagonal $-x_{uv}$, 
which happens precisely when $\sigma$ moves $u$ with $\sigma(u)=v$. 

Let us write $M(\sigma)$ for the set of all incidence assignments 
$\mu:V_r\rightarrow E$ with the properties that (i) each $u\in V_r$ 
fixed by $\sigma$ is assigned to an in-arc $\mu(u)=wu\in E$
for some $w\in V$, and (ii) each $u\in V_r$ moved by $\sigma$ is 
assigned to the unique out-arc $\mu(u)=uv\in E$ with $\sigma(u)=v$.
Let us write $f=f(\sigma)$ for the number of elements in $V_r$ fixed 
by $\sigma$. It is immediate by (i) and (ii) that we have
\begin{equation}
\label{eq:m-sigma}
\prod_{u\in V_r}\ell_{u,\sigma(u)}=
\sum_{\mu\in M(\sigma)}(-1)^{n-1-f(\sigma)}\prod_{u\in V_r}x_{\mu(u)}\,.
\end{equation}
Next observe that from $\mu$ we can reconstruct $\sigma=\sigma(\mu)$ 
by (i) setting $\sigma(u)=u$ for each $u$ assigned to an in-arc in $\mu$, and 
(ii) setting $\sigma(u)=v$ for each $u$ assigned to an out-arc $uv$ in $\mu$.
Thus the union $M=\bigcup_{\sigma\in Sym(V_r)}M(\sigma)$ is disjoint. 
Let us call the elements of $M$ {\em proper} incidence assignments. 
By \eqref{eq:det-lr} and \eqref{eq:m-sigma} we have
\begin{equation}
\label{eq:det-lr-mu}
\det L_r=\sum_{\mu\in M}
(-1)^{n-1-f(\sigma)}\sgn(\sigma(\mu))\prod_{u\in V_r}x_{\mu(u)}\,.
\end{equation}

We claim that an incidence assignment $\mu$ is proper if and only if
for every $u\in V_r$ there is exactly one $u'\in V_r$ such that
$\mu(u')$ is an in-arc to $u$. For the ``only if'' direction, let $\sigma$
be the permutation underlying a proper $\mu$, and observe that
vertices moved by $\sigma$ partition to cycles so a $\sigma$ never moves a 
vertex to a fixed vertex. Thus, we have $u'=u$ for the points fixed by
$\sigma$, and $u'=\sigma^{-1}(u)$ is the vertex preceding $u$ along 
a cycle of $\sigma$ for points moved by $\sigma$. For the ``if'' direction, 
define $\sigma(u)=u$ if $u=u'$ and $\sigma(u')=u$ if $u'\neq u$. In
the latter case we have $\mu(u')=u'u$, which means that $u''\neq u'$
and thus $\sigma(u'')=u'$; by uniqueness of $u'$ eventually a cycle must
close so $\sigma$ is a well-defined permutation underlying $\mu$ and thus
$\mu$ is proper.

Let us write $\mathscr{H}_r$ for the set of all spanning subgraphs of $G$ 
with the property that every vertex in $V_r$ has in-degree $1$ and 
the root $r$ has in-degree $0$.
From the previous claim it follows that we can view the set 
$\mu(V_r)=\{\mu(u):u\in V_r\}$ for a proper $\mu$ as an
element of $\mathscr{H}_r$. Furthermore, $\mu(V_r)$ is connected 
(and hence a spanning out-branching with root $r$) if and only if $\mu(V_r)$ 
is acyclic. 

Consider an arbitrary $H\in\mathscr{H}_r$. If $H$ has a 
cycle, let $C$ be the least cycle in $H$ according to some fixed but
arbitrary ordering of the vertices of $G$. 
(Observe that any two cycles in $H$ must be vertex-disjoint and
cannot traverse $r$ because $r$ has in-degree $0$.)
Now consider an arbitrary proper $\mu$ that realizes $H$ by $\mu(V_r)=H$. 
The cycle $C$ is realized in $\mu$ by either \eqref{eq:in-realization}
(in which case $\sigma(\mu)$ fixes all vertices in $C$), or 
\eqref{eq:out-realization} (in which case $\sigma(\mu)$ traces the cycle $C$).
Furthermore, we may switch between realizations \eqref{eq:in-realization} and 
\eqref{eq:out-realization} so that the number of fixed points in the 
underlying permutation changes by $|V(C)|$ and the sign of the underlying 
permutation gets multiplied by $(-1)^{|V(C)|-1}$. It follows that the 
realizations \eqref{eq:in-realization} and \eqref{eq:out-realization}
have different signs and thus cancel each other in \eqref{eq:det-lr-mu}.
If $H$ does not have a cycle, that is, $H\in\mathscr{B}_r$, it follows 
that there is a unique proper $\mu$ that realizes $H$. 
Indeed, first observe that $H$ can be realized only by assigning in-arcs 
since any assignment of an out-arc in $\mu$ implies a cycle in $H=\mu(V_r)$, 
a contradiction. Second, the in-arcs are unique since each $u\in V_r$ has
in-degree $1$ in $H$. Finally, since $\mu$ assigns only in-arcs the
underlying permutation $\sigma(\mu)$ is the identity permutation which
has $\sgn(\sigma(\mu))=1$ and $(-1)^{n-1-f(\sigma(\mu))}=1$. Thus, each
acyclic $H$ contributes to \eqref{eq:det-lr-mu} through a single
$\mu\in M$ with coefficient $1$. The theorem follows.
\end{proof}

\section{Corollary for $k$-internal out-branchings}

\label{sect:branchings}

This section proves Theorem~\ref{thm:k-internal-out-branchings}.
We rely on a substitution idea of 
Floderus~{\em et~al.}~\cite[Theorem~1]{FloderusLPS15}
to detect monomials with at least $k$ distinct variables. 

Let $G$ be an $n$-vertex directed graph given as input together with 
a nonnegative integer $k$. 
Without loss of generality we may assume that $k\leq n-1$.
Iterate over all choices for a root vertex $r\in V$.
Introduce an indeterminate $y_u$ for each vertex $u\in V$
and an indeterminate $z_{uv}$ for each arc $uv\in E$. 
Introduce one further indeterminate $t$. 
Construct the symbolic Laplacian $L$ of $G$ given 
by \eqref{eq:laplacian} and with the assignment 
$x_{uv}=(1+ty_u)z_{uv}$ to the indeterminate $x_{uv}$ for each $uv\in E$.
Puncture $L$ at $r$ to obtain $L_r$. Using, for example, Berkowitz's 
determinant circuit design~\cite{Berkowitz1984} for an arbitrary
commutative ring with unity, in time $O^*(1)$ build an arithmetic 
circuit $\mathscr{C}$ of size $O^*(1)$ for $\det L_r$. Viewing $\det L_r$ 
as a multivariate polynomial over the polynomial ring 
$R[t,y_u,z_{uv}:u\in V,uv\in E]$ where $R$ is an abstract ring with unity, 
from Theorem~\ref{thm:multivariate-out-branching}
it follows that $G$ has a spanning out-branching rooted at $r$ with at least
$k$ internal vertices if and only if the coefficient of $t^k$ in $\det L_r$ 
(which is a polynomial that is either identically zero 
or both (i) homogeneous of degree $k$ in the 
indeterminates $y_u$ and (ii) homogeneous of degree $n-1$ in 
the indeterminates $z_{uv}$)
has a monomial that is multilinear of degree $k$ in the 
indeterminates $y_u$.
Indeed, observe that the substitution $x_{uv}=(1+ty_u)z_{uv}$
tracks in the degree of the indeterminate $y_u$ whether $u$ occurs as an
internal vertex or not; the indeterminates $z_{uv}$ make sure that distinct
spanning out-branchings will not cancel each other. 

To detect a multilinear monomial in $\mathscr{C}$ 
restricted to the coefficient of $t^k$ we can
invoke~\cite[Lemma~1]{Bjorklund2016} or~\cite[Lemma~2.8]{Koutis2016}.
This results in a randomized algorithm that runs in time $O^*(2^k)$ and has a negligible
probability of reporting a false negative. 
This completes the proof of Theorem~\ref{thm:k-internal-out-branchings}. 
$\qed$


\section{Modular counting of Hamiltonian cycles}

\label{sect:modular}

This section proves Theorem~\ref{thm:counting-modulo-prime-power}.
Fix an arbitrary constant $0<\lambda<1$. 
Let $0<\beta<1/2$ be a constant whose precise value is fixed later. 
Let $p$ be a prime and let $G$ be an $n$-vertex directed graph with 
vertex set $V$ and arc set $E$ given as input.
Without loss of generality (by splitting any vertex $u$ into two vertices, 
$s$ and $t$, with $s$ receiving the out-arcs from $u$, and $t$ 
receiving the in-arcs to $u$) we may count the spanning paths starting 
from $s$ and ending at $t$ instead of spanning cycles.
Similarly, without loss of generality we may assume that $2\leq p<n$.
(Indeed, for $p\geq n$ the counting outcome from 
Theorem~\ref{thm:counting-modulo-prime-power} is trivial.)

\subparagraph*{Sieving for Hamiltonian paths among out-branchings.}
Let $s,t\in V$ be distinct vertices. Let us write $\hp(G,s,t)$ for the
set of spanning directed paths that start at $s$ and end at $t$ in $G$.
Recall that we write $V_t=V\setminus\{t\}$ for the $t$-punctured version
of the vertex set $V$. Let us also write $V_{st}=V\setminus\{s,t\}$.
For $O\subseteq V_t$, let $L_s^O$ be the matrix obtained from the
Laplacian \eqref{eq:laplacian} by first puncturing at $s$ and then 
substituting $x_{uv}=0$ for all arcs $uv\in E$ with $u\in V_t\setminus O$. 
Since a path $P\in \hp(G,s,t)$ is precisely a spanning out-branching rooted 
at $s$ such that every vertex $u\in V_t$ has out-degree $1$, we have, by 
Theorem~\ref{thm:multivariate-out-branching} and the principle of inclusion 
and exclusion, 
\begin{equation}
\label{eq:hamiltonian-path-sieve}
\sum_{P\in\hp(G,s,t)}\prod_{uv\in E(P)}x_{uv}
=\sum_{O\subseteq V_t}(-1)^{|V_t\setminus O|}\det L_s^O\,.
\end{equation}
In particular observe that \eqref{eq:hamiltonian-path-sieve} holds in any
characteristic. 

\subparagraph*{Cancellation modulo a power of $p$.}
With foresight, select $k=\lfloor (1-\lambda)(1/2-\beta)n/p\rfloor$.
Our objective is next to show that by carefully injecting entropy into
the underlying Laplacian we can, in expectation and working modulo $p^k$, 
cancel all but an exponentially negligible fraction of 
the summands on the right-hand side of \eqref{eq:hamiltonian-path-sieve}. 
Furthermore, we can algorithmically narrow down to the nonzero terms,
leading to an exponential improvement to $2^n$. 

Let us assign $x_{uv}=1$ for all $uv\in E$ with 
$u\neq t$. Since no spanning path that ends at $t$ may contain an arc
$tu\in E$ for any $u\in V_t$, we may without loss of generality assume
that $G$ contains all such arcs, and assign, independently and uniformly
at random $x_{tu}\in\{0,1,\ldots,p-1\}$. Thus, the summands $\det L_s^O$
for $O\subseteq V_t$ are now integer-valued random variables and 
\eqref{eq:hamiltonian-path-sieve} evaluates to $|\hp(G,s,t)|$ 
with probability $1$.

Let us next study a fixed $O\subseteq V_t$.
Let $F_O$ be the event that $L_s^O$ has no more than 
$k$ rows where each entry is divisible by $p$. In particular, 
$\det L_s^O\not\equiv 0\pmod{p^k}$ implies $F_O$.
To bound the probability of $F_O$ from above, 
observe that $L_s^O$ is identically zero at each row 
$u\in V_{st}\setminus O$ except possibly at the diagonal entries.
Furthermore, because of the random assignment to the indeterminates $x_{tu}$, 
each diagonal entry at these rows is divisible by $p$ with probability $1/p$. 
Let us take this intuition and turn it into a listing algorithm
for (a superset of the) sets $O\subseteq V_t$ that satisfy $F_O$.

\subparagraph*{Bipartitioning.}
For listing we will employ a meet-in-the-middle approach based on building
each set $O\subseteq V_t$ from two parts using the following bipartitioning. 
Let $V_t^{(1)}\cup V_t^{(2)}=V_t$ be a bipartition with 
$|V_t^{(1)}|=\lceil n/3\rceil$ and $|V_t^{(2)}|=n-1-\lceil n/3\rceil$.
Associate with each $O_1\subseteq V_t^{(1)}$ a vector 
$z^{O_1}\in\{0,1,\ldots,p-1,\infty\}^{V_{st}}$ with the entry at 
$u\in V_{st}$ defined by
\begin{equation}
\label{eq:z1}
z^{O_1}_u=
\begin{cases}
\infty & \text{if $u\in O_1$};\\
\bigl(x_{tu}+\sum_{w\in O_1:wu\in E} x_{wu}\bigr) \bmod{p} & \text{otherwise}.
\end{cases}
\end{equation}
Similarly, associate with each $O_2\subseteq V_t^{(1)}$ a vector 
$z^{O_2}\in\{0,1,\ldots,p-1,\infty\}^{V_{st}}$ with the entry at $u\in V_{st}$ 
defined by
\begin{equation}
\label{eq:z2}
z^{O_2}_u=
\begin{cases}
\infty & \text{if $u\in O_2$};\\
\bigl(-\sum_{w\in O_2:wu\in E} x_{wu}\bigr) \bmod{p} & \text{otherwise}.
\end{cases}
\end{equation}
Suppose now that we have $O_1\subseteq V_t^{(1)}$ and $O_2\subseteq V_t^{(2)}$ 
with $O=O_1\cup O_2$.
We claim that $F_O$ holds only if the vectors $z^{O_1}$ and $z^{O_2}$ 
agree in at most $k$ entries. Indeed, observe that 
$z^{O_1}_u=z^{O_2}_u$ holds only if both $u\in V_{st}\setminus O$
and the $(u,u)$-entry of $L_s^O$ is divisible by $p$. That is, 
$z^{O_1}_u=z^{O_2}_u$ implies the entire row $u$ of $L_s^O$ consists 
only of elements divisible by $p$. 
Thus it suffices to list all pairs $(O_1,O_2)$ such that $z^{O_1}$ 
and $z^{O_2}$ have at most $k$ agreements.

\subparagraph*{Balanced and unbalanced sets.}

To set up the listing procedure, let us now partition the index domain $V_{st}$ 
of our vectors into $b=\lfloor 3\log_2 p\rfloor$ pairwise disjoint sets 
$S_1,S_2,\ldots,S_b$ such that we have 
$\lfloor (n-2)/b\rfloor\leq|S_i|\leq \lceil(n-2)/b\rceil$.

Let us split the sets $O\subseteq V_t$ into two types.
Let us say that $O$ is {\em balanced} if 
$(1/2-\beta)n/b\leq |(V_{st}\setminus O)\cap S_i|\leq (1/2+\beta)n/b$ 
holds for all $i=1,2,\ldots,b$; otherwise $O$ is {\em unbalanced}. 
Recalling that $\sum_{j=0}^\ell\binom{n}{j}\leq 2^{nH(\ell/n)}$ holds for
all integers $1\leq \ell\leq n/2$, where 
$H(\rho)=-\rho\log_2\rho-(1-\rho)\log_2(1-\rho)$ is the binary entropy 
function, observe that there are in total at most 
\begin{equation}
\label{eq:unbalanced-bound}
\begin{split}
2^{n+1-\min_i|S_i|}b\sum_{j=0}^{\lceil(1/2-\beta)n/b\rceil}\tbinom{\lfloor n/b+2\rfloor}{j}
&\leq 2^{n-(n-2)/b+3}2^{(n/b+2)H(1/2-\beta)}b\\
&\leq 2^{n(1-(1-H(1/2-\beta))/b)+7}b
\end{split}
\end{equation}
sets $O$ that are unbalanced. 

\subparagraph*{Precomputation and listing.}

Suppose that $O_1\subseteq V_t^{(1)}$ and $O_2\subseteq V_t^{(1)}$ are 
{\em compatible} in the sense that $z^{O_1}$ and $z^{O_2}$ agree in at 
most $k$ entries. For $S\subseteq V_{st}$ and a vector $z$ whose entries 
are indexed by $V_{st}$, let us write $z_{S}$ for the restriction of $z$ to $S$.
If $O_1$ and $O_2$ are compatible, then by an averaging argument
there must exist an $i=1,2,\ldots,b$ such that $z^{O_1}_{S_i}$ 
and $z^{O_2}_{S_i}$ agree in at most $k/b$ entries. In particular, this 
enables us to iterate over $O_2$ and list all compatible 
$O_1$ by focusing only on each restriction to $S_i$ for 
$i=1,2,\ldots,b$. Furthermore, the search inside $S_i$ can be precomputed
to look-up tables. Indeed, for each $i=1,2,\ldots,b$ and each key 
$g\in\{0,1,\ldots,p-1,\infty\}^{S_i}$, let us build a complete list
of all subsets $O_1\subseteq V_t^{(1)}$ such that $z^{O_1}_{S_i}$ 
and $g$ agree in at most $k/b$ entries. 
These $b$ look-up tables can be built by processing in total at most
\[
 \sum_{i=1}^b 2^{V_t^{(1)}}(p+1)^{|S_i|}
 \leq 2^{n/3+7}2^{(n/(\lfloor 3\log_2 p\rfloor)+2)\log_2(p+1)}\log_2 p
 = O(2^{0.87n})
\]
pairs $(O_1,g)$. This takes time $O^*(2^{0.87n})$ in total.

The main listing procedure now considers each $O_2\subseteq V_t^{(2)}$ 
in turn, and for each $i=1,2,\ldots,b$ consults the look-up table 
for direct access to all $O_1$ such that $z^{O_1}_{S_i}$ and 
$z^{O_2}_{S_i}$ agree in at most $k/b$ entries. In particular this will
list all compatible pairs $(O_1,O_2)$ and hence all 
sets $O=O_1\cup O_2$ such that $F_O$ holds. 

\subparagraph*{Expected running time.}
Let us now analyze the expected running time of the algorithm. 
We start by deriving an upper bound for the expected number of pairs 
$(O_1,O_2)$ considered by the main listing procedure.
First, observe that the total number of pairs $(O_1,O_2)$ considered
by the procedure with $O=O_1\cup O_2$ unbalanced is bounded from
above by our upper bound \eqref{eq:unbalanced-bound} for the total number
of unbalanced $O$. Indeed, $O_1=O\cap V_t^{(1)}$ and $O_2=O\cap V_t^{(2)}$ 
are uniquely determined by $O$.

Next, for a pair $(O_1,O_2)$ with balanced $O=O_1\cup O_2$ 
and $i=1,2,\ldots,b$, let $G_{O_1,O_2,i}$ be 
the event that $z^{O_1}_{S_i}$ and $z^{O_2}_{S_i}$ agree in at most $k/b$ 
entries. We seek an upper bound for the probability of $G_{O_1,O_2,i}$ to
obtain an upper bound for the expected number of pairs with balanced 
$O=O_1\cup O_2$ considered by the main listing procedure. 
Let $A_{O_1,O_2,i}$ be the number of entries in which $z^{O_1}_{S_i}$ 
and $z^{O_2}_{S_i}$ agree. 
We observe that $A_{O_1,O_2,i}$ is binomially distributed with expectation 
$|(V_t\setminus O)\cap S_i|/p$. Since $O$ is balanced, we have 
$(1/2-\beta)n/b\leq|(V_t\setminus O)\cap S_i|\leq (1/2+\beta)n/b$. 
We also recall that $k=\lfloor (1-\lambda)(1/2-\beta)n/p\rfloor$.
A standard Chernoff bound now gives
\[
\begin{split}
\Pr(G_{O_1,O_2,i})
&\leq\Pr\bigl(A_{O_1,O_2,i}\leq k/b\bigr)\\
&\leq\Pr\bigl(A_{O_1,O_2,i}\leq (1-\lambda)|(V_t\setminus O)\cap S_i|/p\bigr)\\
&\leq\exp\bigl(-\lambda^2|(V_t\setminus O)\cap S_i|/(2p)\bigr)\\
&\leq\exp\bigl(-\lambda^2(1/2-\beta)n/(2pb)\bigr)\,.
\end{split}
\]
Recalling that $b=\lfloor 3\log_2 p\rfloor$, the main listing procedure 
thus considers in expectation at most
$2^n\exp\bigl(-\lambda^2(1/2-\beta)n/(2p(3\log_2 p)\bigr)$
pairs $(O_1,O_2)$ with $O=O_1\cup O_2$ balanced. 
Recalling our upper bound for the total number of unbalanced sets 
\eqref{eq:unbalanced-bound}, we thus have that the main listing procedure runs 
in $O^*\bigl(%
2^n\exp\bigl(-\lambda^2(1/2-\beta)n/(2p(3\log_2 p))\bigr)+
2^{n(1-(1-H(1/2-\beta))/(3\log_2 p))}\bigr)$ expected time. 
Recalling that precomputation runs in $O^*(2^{0.87n})$ time, 
we thus have for $\beta=1/6$ 
that the entire algorithm runs in $O^*(2^{n(1-\lambda^2/(19p\log_2 p))})$ 
expected time and computes $|\hp(G,s,t)|$ modulo 
$p^{\lfloor(1-\lambda)n/(3p)\rfloor}$.
This completes the proof of Theorem~\ref{thm:counting-modulo-prime-power}. 
$\qed$

\section{Directed Hamiltonicity via quasi-Laplacian determinants}

\label{sect:quasi-laplacian}

This section proves Theorem~\ref{thm:directed-hamiltonicity}.
Let $G$ be a directed $n$-vertex graph given as input.

\subparagraph*{Finding a maximum independent set.}
Let $B\cup Y=V(G)$ be a partition of the
vertex set into two disjoint sets $B$ (``blue'') and $Y$ (``yellow'') such 
that no arc has both of its ends in $Y$. That is, $Y$ is an 
independent set. 

We can find an $Y$ of the maximum possible size as follows. 
First, in time polynomial in $n$ compute the maximum-size matching in 
the undirected graph obtained from $G$ by disregarding the orientation of 
the arcs. This maximum-size matching must consist of at least 
$\lfloor n/2\rfloor$ edges or $G$ does not admit a Hamiltonian cycle. 
(Indeed, from a Hamiltonian
cycle we can obtain a matching with $\lfloor n/2\rfloor$ edges by taking
every other arc in the cycle.) Since for each matching edge it holds that
both ends of the edge cannot be in an independent set, we can in time
$O^*(3^{n/2})$ find a maximum-size independent set $Y$ of $G$. 
Furthermore, $\alpha(G)=|Y|\leq \lfloor n/2\rfloor+1$, so we are 
within our budget of $O^*(3^{n(G)-\alpha(G)})$ in terms of running time.
In fact, $|Y|\leq\lfloor n/2\rfloor$ or otherwise $G$ trivially does 
not admit a Hamiltonian cycle.

\subparagraph*{The symbolic quasi-Laplacian.}
We will first define the quasi-Laplacian and then give intuition for its
design. Let us work over a field of characteristic $2$. 
For each $y\in Y$ introduce a copy $y_\inin$ and let $Y_\inin$ be the set
of all such copies. Similarly, for each $y\in Y$ introduce a copy $y_\outin$ 
and let $Y_\outin$ be the set of all such copies. We assume that $Y_\inin$
and $Y_\outin$ are disjoint. For each $j\in Y_\inin\cup Y_\outin$ let us 
write $\underline{j}\in Y$
for the underlying element of $Y$ of which $j$ is a copy.
Let $B_*$ be a set of $n-2|Y|$ elements that is 
disjoint from both $Y_\inin$ and $Y_\outin$. 
For each $uv\in E$ and each $j\in B_*\cup Y_\inin\cup Y_\outin$, 
introduce an indeterminate $x_{uv}^{(j)}$.

Select an arbitrary vertex $s\in B$ for purposes of breaking symmetry
and let $I,O\subseteq B$. 
The {\em quasi-Laplacian} 
$Q^{I,O,s}=Q^{I,O,s}(G)$ of $G$ with {\em skew at} $s$ be the $n\times n$ 
matrix whose rows are indexed by 
$u\in B\cup Y$ 
and whose columns are indexed by 
$j\in B_*\cup Y_\inin\cup Y_\outin$
such that the $(u,j)$-entry is defined by
\begin{equation}
\label{eq:quasi-laplacian}
q^{I,O,s}_{uj}=
\begin{cases}
\sum_{w\in O:wu\in E,u\in I}x_{wu}^{(j)}&\\
\ \ \ +\sum_{w\in I:uw\in E,u\in O\setminus\{s\}}x_{uw}^{(j)}
& \text{(a) if $u\in B$ and $j\in B_*$};\\
x_{u\underline{j}}^{(j)}
& \text{(b) if $u\in O\setminus\{s\}$ and $j\in Y_\inin$ with $u\underline{j}\in E$};\\
x_{\underline{j}u}^{(j)}
& \text{(c) if $u\in I$ and $j\in Y_\outin$ with $\underline{j}u\in E$};\\
\sum_{w\in O:wu\in E}x_{wu}^{(j)}
& \text{(d) if $u\in Y$ and $j\in Y_\inin$ with $u=\underline{j}$};\\
\sum_{w\in I:uw\in E}x_{uw}^{(j)}
& \text{(e) if $u\in Y$ and $j\in Y_\outin$ with $u=\underline{j}$};\\
0
& \text{otherwise}.
\end{cases}
\end{equation}
Let us next give some intuition for \eqref{eq:quasi-laplacian} before
proceeding with the proof. 

Analogously to the Laplacian \eqref{eq:laplacian}, 
the quasi-Laplacian \eqref{eq:quasi-laplacian} has been designed so 
that the monomials of each row $u\in B\cup Y$ of $Q^{I,O,s}$ control 
the assignment of either an in-arc or an out-arc to $u$ in an incidence
assignment, and the skew at $s$ is used to break symmetry so that $s$ is 
always assigned an in-arc to $s$. In particular, without the skew at 
row $s$ and with $I=O$, each column of $Q^{I,O,s}$ would sum to zero, 
in analogy with the (non-punctured) Laplacian. 


Let us now give intuition for the columns $j\in Y_\inin$ and $j\in Y_\outin$.
First, observe by (b) and (d) in \eqref{eq:quasi-laplacian} that
selecting a monomial from column $j\in Y_\inin$ corresponds to making sure that 
the in-degree of $\underline j$ is $1$. Such a selection may be either 
a ``quasi-diagonal'' assignment of the in-arc $w\underline{j}$ to 
$u=\underline j\in Y$ via (d) for some $w\in B$; 
or an ``off-diagonal'' assignment of the out-arc $u\underline{j}$ 
to $u\in B$ via (b). 
Second, observe by (c) and (e) in \eqref{eq:quasi-laplacian} that
selecting a monomial from column $j\in Y_\outin$ corresponds to making 
sure that the out-degree of $\underline j$ is $1$.
{\em Thus, the columns $j\in Y_\inin\cup Y_\outin$ enable us to
make sure that an incidence assigment has both in-degree $1$ and out-degree
$1$ at each $u\in Y$ without the use of sieving. This gives us the speed-up
from $O^*(3^n)$ to $O^*(3^{n-|Y|})$ running time.} Observe also that the
structure for the quasi-Laplacian $Q^{I,O,s}$ is enabled precisely because
$Y$ is an independent set and thus no arc contributes to both in-degree
and out-degree in $Y$.

\subparagraph*{The quasi-Laplacian determinant sieve.}
Recalling that we are working over a field of characteristic 2, 
let us study the sum
\begin{equation}
\label{eq:main-sieve}
\sum_{\substack{I,O\subseteq B\\I\cup O=B\\s\in I}}
\det Q^{I,O,s}
=
\sum_{\substack{I\subseteq B}}
\sum_{\substack{O\subseteq B}}
\det Q^{I,O,s}
=\!\!\!\!
\sum_{\substack{\sigma:B\cup Y\rightarrow B_*\cup Y_\inin\cup Y_\outin\\\text{$\sigma$ bijective}}}
\ 
\sum_{\substack{I\subseteq B}}
\sum_{\substack{O\subseteq B}}
\prod_{u\in B\cup Y}q^{I,O,s}_{u,\sigma(u)}\,.
\end{equation}
Observe that the first equality in \eqref{eq:main-sieve} holds 
because $Q^{I,O,s}$ has by \eqref{eq:quasi-laplacian} an identically zero 
row unless $I\cup O=B$ and $s\in I$; the second equality holds by definition
of the determinant in characteristic $2$ and changing the order of summation.
From \eqref{eq:quasi-laplacian} and the right-hand side of 
\eqref{eq:main-sieve} it is immediate that \eqref{eq:main-sieve} is either 
identically zero or a homogeneous polynomial of degree $n$ in the $n|E|$ 
indeterminates $x_{uv}^{(j)}$ for $j\in B_*\cup Y_\inin\cup Y_\outin$ 
and $uv\in E$.%
\footnote{%
Our algorithm for deciding Hamiltonicity will naturally not work with 
a symbolic representation of \eqref{eq:main-sieve} but rather in 
a homomorphic image under a random evaluation homomorphism.}
We claim that \eqref{eq:main-sieve} is not identically zero if and only if
$G$ admits at least one spanning cycle. Furthermore, each spanning cycle 
in $G$ defines precisely $|B_*|!$ distinct monomials in \eqref{eq:main-sieve}.

To establish the claim, fix a bijection 
$\sigma:B\cup Y\rightarrow B_*\cup Y_\inin\cup Y_\outin$. 
Let us write $M(\sigma)$ for the set of all incidence
assignments $\mu:B\cup Y\rightarrow E$ that are {\em proper} in 
the sense that all of the following six requirements hold 
(cf.~\eqref{eq:quasi-laplacian}):
\begin{itemize}
\item (s): $\mu(s)$ is an in-arc to $s$;
\item (a): for all $u\in B$ with $\sigma(u)\in B_*$ it holds that $\mu(u)$ has
both of its vertices in $B$;
\item (b): for all $u\in B$ with $\sigma(u)\in Y_\inin$ it holds that $\mu(u)$
is an in-arc to $\underline{\sigma(u)}$;
\item (c): for all $u\in B$ with $\sigma(u)\in Y_\outin$ it holds that $\mu(u)$
is an out-arc from $\underline{\sigma(u)}$;
\item (d): for all $u\in Y$ with $\sigma(u)\in Y_\inin$ it holds that
$\mu(u)$ is an in-arc to $u$ and $u=\underline{\sigma(u)}$; and 
\item (e): for all $u\in Y$ with $\sigma(u)\in Y_\outin$ it holds that
$\mu(u)$ is an out-arc from $u$ and $u=\underline{\sigma(u)}$.
\end{itemize}

Observe that each $\mu\in M(\sigma)$ defines a collection of $n$ arcs 
$\mu(B\cup Y)=\{\mu(u):u\in B\cup Y\}$. Let us write 
$Z_\inin^\mu$ (respectively, $Z_\outin^\mu$) 
for the set of vertices in $B\cup Y$ with zero in-degree (out-degree) with
respect to the arcs in $\mu(B\cup Y)$. Since $\sigma$ is a bijection
and thus has a preimage for each $j\in Y_\inin\cup Y_\outin$, 
from (b,c,d,e) above it follows that $Z_\inin^\mu\subseteq B$ and
$Z_\outin^\mu\subseteq B$. Furthermore, from (a,b,c,d,e) it follows
that for the arcs in $\mu(B\cup Y)$
the sum of the in-degrees (and the sum of the out-degrees) of the
vertices in $B$ is $|B|=|B_*|+|Y|$. Thus, we have that 
$Z_\inin^\mu$ and $Z_\outin^\mu$ are both empty if and only if 
for the arcs $\mu(B\cup Y)$ both the in-degree and the out-degree of 
every vertex $u\in B\cup Y$ is $1$. (Note that the claim is immediate
for $u\in Y$ by (b,c,d,e) and bijectivity of $\sigma$.)

Let us now study the right-hand side of \eqref{eq:main-sieve} for
a fixed $\sigma$. Using (a,b,c,d,e) and \eqref{eq:quasi-laplacian} 
to rearrange in terms of incidence assignments, we have 
\begin{equation}
\label{eq:sigma-i-o-z-sum}
\sum_{I\subseteq V}\sum_{O\subseteq V}
\prod_{u\in B\cup Y}q^{I,O,s}_{u,\sigma(u)}
=
\sum_{\mu\in M(\sigma)}
\prod_{u\in B\cup Y}x_{\mu(u)}^{(\sigma(u))}
\sum_{I\subseteq Z_\inin^\mu}
\sum_{O\subseteq Z_\outin^\mu}
1\,.
\end{equation}
Since we are working in characteristic $2$, all other $\mu\in M(\sigma)$ 
except those for which $\mu(B\cup Y)$ is a cycle cover will cancel in 
the right-hand side of \eqref{eq:sigma-i-o-z-sum}. 

Take the sum of \eqref{eq:sigma-i-o-z-sum} over all bijections $\sigma$.
Consider an arbitrary cycle cover of $B\cup Y$. Let $C$ be a cycle in this
cycle cover. Assuming that $C$ does not contain $s$, we can realize
$C$ in an incidence assignment $\mu:B\cup Y\rightarrow E$ either using
\eqref{eq:in-realization} or \eqref{eq:out-realization}. If $C$ contains
$s$ and $\mu$ is proper, only the realization \eqref{eq:in-realization} 
is possible by (s). To see that realization with a proper $\mu\in M(\sigma)$ 
for some $\sigma$ is possible, consider an arbitrary $u\in B\cup Y$ and 
observe that each image $\mu(u)$ by (b,c,d,e) uniquely determines the 
image $\sigma(u)\in Y_\inin\cup Y_\outin$ when $\mu(u)$ has one vertex in 
$Y$; when $\mu(u)$ has both vertices in $B$, an unused 
$\sigma(u)\in B_*$ may be chosen arbitrarily so that $\mu\in M(\sigma)$. 
It follows that any cycle cover with $c$ cycles is realized 
as exacly $|B_*|!$ distinct monomials 
$\prod_{u\in B\cup Y}x_{\mu(u)}^{(\sigma(u))}$ in \eqref{eq:main-sieve}, 
each with coefficient $2^{c-1}$. This coefficient is nonzero 
if and only if $c=1$.

\subparagraph*{Completing the algorithm.}
To detect whether the given $n$-vertex directed graph $G$ 
admits a Hamiltonian cycle, first decompose 
the vertex set into disjoint $V=B\cup Y$ with $Y$ an independent set
of size $|Y|=\alpha(G)$ using the algorithm described earlier. 
Next, in time $O^*(1)$ construct an irreducible polynomial of degree 
$2\lceil\log_2 n\rceil$ over $\mathbb{F}_2$ (see e.g.~von zur Gathen
and Gerhard~\cite[\S14.9]{vonzurGathen2013}) to enable arithmetic in
the finite field of order $q=2^{2\lceil\log_2 n\rceil}\geq n^2$ in
time $O^*(1)$ for each arithmetic operation. 
Next, assign an independent uniform random value from $\mathbb{F}_q$ 
to each indeterminate $x_{uv}^{(j)}$ with $j\in B_*\cup Y_\inin\cup Y_\outin$
and $uv\in E$. Finally, using the assigned values for the indeterminates,
compute the left-hand side of \eqref{eq:main-sieve} using, for example,
Gaussian elimination to compute each determinant $\det Q^{I,O,s}$ 
in $O^*(1)$ operations in $\mathbb{F}_q$. Let us write $r\in\mathbb{F}_q$
for the result of this computation. In particular, we can compute $r$
from a given $G$ in total $O^*(3^{|B|})=O^*(3^{n(G)-\alpha(G)})$ operations
in $\mathbb{F}_q$, and consequently in total $O^*(3^{n(G)-\alpha(G)})$ time.
If \eqref{eq:main-sieve} is identically zero, then clearly $r=0$ 
with probability $1$. If \eqref{eq:main-sieve} is not identically zero
(and hence a homogeneous polynomial of degree $d=n$ 
in the indeterminates) then by the DeMillo--Lipton--Schwartz--Zippel 
lemma~\cite{DeMilloLipton1978,Schwartz1980,Zippel1979} we have $r\neq 0$ 
with probability at least $1-d/q\geq 1-n/n^2\geq 1-1/n=1-o(1)$. 
Thus we can decide whether $G$ is Hamiltonian based on whether $r\neq 0$.
In particular this gives probability $o(1)$ of reporting a false negative. 
This completes the proof of Theorem~\ref{thm:directed-hamiltonicity}. $\qed$

\subparagraph*{Acknowledgements.}

The research leading to these results has received 
funding from the Swedish Research Council grants VR 2012-4730 ``Exact 
Exponential-Time Algorithms'' and VR 2016-03855 ``Algebraic Graph Algorithms'' (A.B.), the 
European Research Council under the European Union's Seventh Framework 
Programme (FP/2007-2013) / ERC Grant Agreement 338077 ``Theory and 
Practice of Advanced Search and Enumeration'' (P.K.), and grant
NSF CAREER CCF-1149048 (I.K.).
Work done in part while the authors were at Dagstuhl Seminar 17041 in 
January 2017 and at the Simons Institute for the Theory of Computing 
in December 2015.

\newpage
\begin{center}
	{\LARGE \sffamily\bfseries APPENDIX}
\end{center}

\appendix

\section{Detecting a $k$-Leaf Out-Branching}

\label{appendix:k-leaf}

The {\em $k$-Leaf Out-Branching} problems asks, given an $n$-vertex directed
graph $G$ and an integer $k$ as input, whether $G$ has a 
spanning out-branching with at least $k$ leaves. Equivalently, we are asked 
to detect a spanning out-branching with at most $n-k$ internal vertices. 

Let us algebraize the $k$-Leaf problem by pursuing an analogy with 
the approach in Sect.~\ref{sect:branchings}. Without loss of generality
we may assume that $1\leq k\leq n$ and $n\geq 2$. 
Iterate over all choices for a root vertex $r\in V$. Since
$n\geq 2$ we have that $r$ must be an internal vertex in any out-branching
rooted at $r$. Introduce an indeterminate $y_u$ for each vertex $u\in V$.
Construct the symbolic Laplacian $L$ of $G$ given 
by \eqref{eq:laplacian} and with the assignment 
$x_{uv}=y_u$ to the indeterminate $x_{uv}$ for each $uv\in E$.
Puncture $L$ at $r$ to obtain $L_r$. Viewing $y_r \det L_r$ 
as a multivariate polynomial over the indeterminates $y_u$, 
from Theorem~\ref{thm:multivariate-out-branching}
it follows that $G$ has a spanning out-branching rooted at $r$ with 
at most $n-k$ internal vertices if and only if 
$y_r\det L_r$ (which is a polynomial that is either identically zero 
or homogeneous of degree $n$ in the 
indeterminates $y_u$) has a monomial with at most $n-k$ 
of the $n$ indeterminates $y_u$.

This motivates the following parameterized monomial detection problem. 
Namely, the $(n-k)$-{\em DV} problem asks, given as input a polynomial 
$P\in \mathbb{Z}[y_1,y_2,\ldots,y_n]$ that is homogeneous of degree $n$ 
and with nonnegative integer coefficients, whether $P$ has 
a monomial with at most $n-k$ distinct indeterminates.

\medskip

The $k$-DV problem was discussed in~\cite{FloderusLPS15} and it was 
shown that it is W[2]-hard via a reduction from the parameterized hitting 
set problem. However, the $(n-k)$-DV problem is fixed-parameter tractable 
with respect to $k$. We prove the following:
\begin{theorem}[Detecting a monomial with few indeterminates] \label{th:dvm}
	There is a randomized algorithm that decides the $(n-k)$-DV problem 
        using $O^*(2^{k+s})$ evaluations of $P$
	and with negligible probability for a false negative, where 
	$s\leq k$ is the minimum number of indeterminates with degree greater 
        than $1$ taken over all monomials of $P$ with at most $n-k$ distinct 
        indeterminates. 
\end{theorem}
\begin{proof}
We introduce a set of extra indeterminates 
$W = w_1,w_2,\ldots,w_n$, and two special
indeterminates $a,b$. We will substitute these
two indeterminates into $X$ and $Y$ and compute
a homogeneous polynomial in $\{a,b\}$.  
The computation can be done either symbolically or with a sufficient
number of numerical evaluations and interpolation. 

The algorithm iterates the following two steps $O^*(2^{k+s})$ times: 
\begin{enumerate}
\item
For each $i=1,2,\ldots,n$ independently, substitute: \\
either $(y_i\leftarrow a, w_i\leftarrow b)$ or 
$(y_i\leftarrow b, w_i\leftarrow a)$, each with probability $1/2$.
\item
Perform the random substitution on the polynomial
$P(y_1,y_2,\ldots,y_n)w_1w_2\cdots w_n$ to obtain a homogeneous polynomial
of degree $2n$ in $a$ and $b$. If the result is {\em not} divisible 
by $(ab)^{n-k+1}$ then halt with output YES.
\end{enumerate}
If none of the iterations halts, then halt with output NO.	

Let us now analyze the algorithm. First, consider the effect of the substitution to a monomial with at least $n-k+1$ of the indeterminates $y_1,y_2,\ldots,y_n$.Since the monomial gets multiplied by $w_1w_2\cdots w_n$, the substitution yields a monomial in $a,b$ that is a multiple of $(ab)^{n-k+1}$. Thus, if all monomials of $P$ have at least $n-k+1$ indeterminates then by linearity the result of substitution to $P(y_1,y_2,\ldots,y_n)w_1w_2\cdots w_n$ is a multiple of $(ab)^{n-k+1}$. Thus, the algorithm halts with output NO.

Next, consider a monomial with exactly $n-k'$ indeterminates, among which exactly $s$ have degree at least $2$. Because the monomial has degree $n$, we have $s\leq k'$.
The monomial is thus of the form
\[
y_{i_1}^{d_1}y_{i_2}^{d_2}\cdots y_{i_s}^{d_s}y_{i_{s+1}}\cdots y_{i_{n-k'}}
\]
with $d_1,d_2,\ldots,d_s\geq 2$. 
Multiplying with $w_1w_2\cdots w_n$ we thus have
\[
\left(\prod_{t=1}^{n-k'} (y_{i_t} w_{i_t})\right)  \left(y_{i_1}^{d_1-1} \cdots y_{i_s}^{d_s-1}\right) \left(\prod_{t=n-k'+1}^n w_{i_t} \right)\,.
\]
Under the random substitution the first product evaluates to $(ab)^{n-k'}$. Now note that each of the remaining indeterminates in the product has randomly and independently been assigned to $a$ or $b$, with probability $1/2$. There are $s+k'$ such indeterminates. It follows that with probability at least $2^{-(k'+s)}$ the monomial becomes $(ab)^{n-k'} a^{k}$ or $(ab)^{n-k'} b^{k'}$. Thus, when $k'\geq k$ the resulting monomial is not divisible by $(ab)^{n-k+1}$ with probability at least $2^{-(k+s)}$. Furthermore, by linearity the result of the substitution to $P(y_1,y_2,\ldots,y_n)w_1w_2\cdots w_n$ is not divisible either since $P$ has nonnegative coefficients and thus cancellations are not possible. Repeating $O^*(2^{k+s})$ times renders negligible the probability of a false negative.
\end{proof}

We thus have the following:
\begin{theorem}[Detecting a $k$-leaf out-branching]
	The $k$-Leaf Out-Branching problem can be decided with a randomized algorithm in $O^*(2^{k+s})$ time, where $s$ is the minimum number of internal vertices with out-degree greater than 1 taken over all out-branchings with at least $k$ leaves. The algorithm has a negligible probability of reporting a false negative.
\end{theorem}
\begin{proof} 
Consider each root vertex $r\in V$ in turn and let $P$ be the 
polynomial $y_r\det L_r$ introduced in the beginning of this section.
For any fixed out-branching, the internal vertices with out-degree greater than $1$ correspond to variables with degree greater than $1$ in the associated 
monomial. The theorem follows by a direct application of Theorem~\ref{th:dvm}.
\end{proof}

\subparagraph*{Note.} As observed by M.~Koivisto, the algorithm can be somewhat improved, by skewing the probabilities in the substitution $(y_i\leftarrow a, w_i\leftarrow b)$ or $(y_i\leftarrow b, w_i\leftarrow a)$ to $p$ and $1-p$ respectively, where $p=k/(k+s)$. Using a very similar argument one can show that the problem can be decided with $O^*(c^{k+s})$ evaluations of $P$, where $c^{-1}=p^p(1-p)^{1-p}$. For example, when $s=k/2$ the running time becomes $O^*(1.89^{s+k})$, and when $s$ is some fixed constant (that is, when a tree has a small number of internal branching vertices) then the algorithm runs in polynomial time. The worst-case running time of the algorithm is $O^*(2^n)$ when $k=s=n/2$, but it gets smaller when $k>n/2$, nearly matching exhaustive search at the W[2]-hard side of the range, that is, when $n-k$ is constant. 

\medskip
The best known algorithm for the directed $k$-leaf problem is based on a branch-and-bound approach and  runs in $O^*(3.72^k)$ time~\cite{Daligault2010}. Our algorithm runs in $O^*(4^k)$ time in the worst case, but it can adapt to properties of the input and actually be faster for instances containing a $k$-leaf tree with at most $10k/11$ internal vertices of out-degree greater than 1.

\section{Corollaries and variants of Theorem~\ref{thm:counting-modulo-prime-power}}

\label{appendix:counting-modulo-prime-power}

This appendix proves corollaries and variants
of Theorem~\ref{thm:counting-modulo-prime-power} to demonstrate 
the versatility of the approach.

\subparagraph*{Boosting the modulus for counting by Chinese remaindering.}
From number theory we know
(Merten's First Theorem, cf.~Tenenbaum~\cite[Theorem~7]{Tenenbaum1995}) that
if we take the sum over the primes $p$ at most $q>2$, we have
\[
\sum_{p\leq q} \frac{\ln p}{p}\geq \ln q - 2\,.
\]
Thus, using Theorem~\ref{thm:counting-modulo-prime-power} for all 
primes $p$ at most $q$, from the Chinese Remainder
Theorem we obtain that there exists a constant $c>0$ such for all 
large enough integers $q$ we can in expected time 
$O^*(2^{n(1-cn/(q\log q))})$ compute the number of Hamiltonian 
cycles in a given $n$-vertex directed graph modulo an integer $M$
that satisfies
\[
M \geq \prod_{p\leq q}p^{\lfloor n/(3p)\rfloor}
\geq 
   \exp\biggl(\frac{1}{4}n\sum_{p\leq q}\frac{\ln p}{p}\biggr)
\geq
   (e^{-2}q)^{n/4}\,.
\]

\subparagraph*{Counting Hamiltonian cycles when there are few.}
The following result is an immediate corollary of modulus-boosting:

\begin{theorem}[Counting Hamiltonian cycles when there are few]
\label{thm:d-pow-n}
There is a randomized algorithm and a constant $c>0$ such that for all 
constants $d>1$ the following holds: when given an $n$-vertex directed 
graph $G$ with at most $d^n$ Hamiltonian cycles as input, the algorithm
computes the exact number of Hamiltonian cycles in $G$ in 
$O^*(2^{n(1-c/d^5)})$ expected time.
\end{theorem}
\begin{proof}
Boost the modulus with $q=e^2d^{4}$ to observe that $M>d^n$
for all large enough $n$.
\end{proof}

\begin{theorem}[Counting Hamiltonian cycles with bounded average out-degree]
There is a randomized algorithm and a constant $c>0$ such that for all 
constants $d>1$ the following holds: when given an $n$-vertex directed 
graph $G$ with average out-degree at most $d$, the algorithm computes the 
exact number of Hamiltonian cycles in $G$ in $O^*(2^{n(1-c/d^5)})$ 
expected time.
\end{theorem}
\begin{proof}
For each vertex $u\in V(G)$, let $d_u$ be the out-degree of $u$.
We have that $G$ has at most $\prod_{u\in V(G)} d_u$ Hamiltonian
cycles. On the other hand, the average degree of $G$ is 
$\frac{1}{n}\sum_{u\in V} d_u$. From the inequality between arithmetic 
and geometric means it follows that there are at most $d^n$ Hamiltonian
cycles in $G$, and Theorem~\ref{thm:d-pow-n} applies.
\end{proof}

\bibliographystyle{plain}

\end{document}